\documentclass[final,1p,times]{elsarticle}
\journal{E}
\usepackage{tikz-cd}
\usepackage{extpfeil}
\usepackage{amssymb}
\usepackage{amsthm}
\usepackage{latexsym}
\usepackage{amsmath}
\usepackage{color}
\usepackage{graphicx}
\usepackage{indentfirst}
\usepackage{mathrsfs}
\usepackage{lipsum}
\usepackage{tabularx}
\usepackage{booktabs}
\usepackage{geometry}
\usepackage{comment}
\allowdisplaybreaks

\newtheorem{theorem}{\color{black}\indent \textbf{Theorem}}[section]

\newtheorem{proposition}{\color{black}\indent Proposition}[section]
\newtheorem{definition}{\color{black}\indent Definition}[section]
\newtheorem{remark}{\color{black}\indent Remark}[section]

\newtheorem{example}{\color{black}\indent Example}[section]

\begin{document}
	
	\begin{frontmatter}
		\title{{(Quasi)} Hamiltonian Systems and Non-Decomposable Poisson Geometry}
		
\author{
Cristina Sardón$^{a}$,
Xuefeng Zhao$^{b}$ \\[1ex]
$^{a}$Department of Applied Mathematics, Universidad Politécnica de Madrid, 
Av. Juan de Herrera 6, 28040, Madrid, Spain \\
$^{b}$College of Mathematics, Jilin University, Changchun, 130012, P. R. China \\[1ex]
 \texttt{mariacristina.sardon@upm.es}, \texttt{zhaoxuef@jlu.edu.cn}
}

		\begin{abstract}
			In this work, we conduct a systematic study of Hamiltonian and quasi-Hamiltonian systems within the framework of non-decomposable generalized Poisson geometry. Our focus lies on the interplay between the algebraic structure of non-decomposable generalized Poisson brackets and the dynamical behavior of systems exhibiting specific symmetry properties. In particular, we demonstrate that if a dynamical system admits suitable invariance conditions—such as those arising from Lie symmetries or conserved quantities—it can be formulated as a quasi-Hamiltonian system, or even as a genuinely Hamiltonian system, with respect to a suitably constructed non-decomposable generalized Poisson structure. This result offers a unified geometric framework for analyzing such systems and underscores the capacity of non-decomposable generalized Poisson structures in contexts involving multi-Hamiltonian or higher-order dynamics.
		\end{abstract}
		
	\end{frontmatter}
%	\tableofcontents

		\section{Introduction}
		
		The incorporation of additional compatible geometric structures plays a fundamental role in the modern geometric approach to the study of dynamical systems via vector fields. These structures not only enrich the geometric formulation but also offer powerful tools for uncovering hidden symmetries and conserved quantities. As a result, there has been sustained interest in identifying and classifying such structures, particularly those that can lead to alternative Hamiltonian descriptions of a system.
		
		A notable contribution in this direction was made by Hojman~\cite{Hojman1992}, who introduced a general technique that constructs an admissible Hamiltonian structure for a given equation of motion using only an infinitesimal symmetry transformation and a conserved quantity. Remarkably, this method is applicable to both ordinary and partial differential equations, making it highly versatile. Subsequent developments~\cite{Hojman1996,Hojman2002} extended Hojman’s method to encompass field-theoretic systems, even in the absence of a Lagrangian formulation—thereby offering a symmetry-based geometric framework that bypasses the traditional variational approach. For a comprehensive overview of recent advances stemming from Hojman’s idea, we refer the reader to~\cite{Hojman2020} and references therein.
		
		Historically, the geometric formulation of classical mechanics began with the use of symplectic geometry, which provided a natural setting for Hamiltonian dynamics. However, to accommodate more general physical and mathematical scenarios—including systems with constraints or degeneracies—this framework was extended to include presymplectic and Poisson manifolds. In 1973, Nambu~\cite{Nambu1973} proposed an even broader generalization by introducing a multi-bracket formalism, aimed at describing dynamics in higher-dimensional phase spaces.
		
		In particular, Nambu considered systems defined on a three-dimensional phase space with coordinates \( (x_1, x_2, x_3) \), and introduced a ternary bracket for functions \( f_1, f_2, f_3 \) defined as:
		\[
		\{f_1, f_2, f_3\} = \frac{\partial(f_1, f_2, f_3)}{\partial(x_1, x_2, x_3)},
		\]
		where the right-hand side denotes the Jacobian determinant. This so-called \emph{Nambu bracket} allows the dynamics of any observable \( f \) to be written in the form:
		\[
		\frac{df}{dt} = \{f, h_1, h_2\},
		\]
		where \( h_1 \) and \( h_2 \) are now two ``Hamiltonian'' (or ``Nambu'') functions, both of which must be conserved. This framework not only generalizes the classical Poisson bracket (recovered when one Hamiltonian is fixed) but also naturally accommodates systems with multiple conserved quantities. The relation between Nambu mechanics and conventional Hamiltonian mechanics was soon investigated by several authors~\cite{NambuStudies1,NambuStudies2,NambuStudies3}, paving the way for further mathematical developments.
		
		Later, Takhtajan~\cite{Takhtajan1994} gave a rigorous mathematical foundation to this theory by introducing the concept of \emph{Nambu--Poisson structures}, based on an axiomatic formulation of the \( m \)-ary bracket. His formulation triggered a wave of research on higher-order brackets and their algebraic and geometric properties~\cite{Guha,Guha2,Hietarinta,IbanezR,Morando,Pandit}. Among these developments, one of particular interest is the discovery of a link between Nambu--Poisson manifolds and \emph{Leibniz algebroids}~\cite{Hagiwara,IbanezR}, which provide a natural generalization of Lie algebroids suited for handling non-binary operations. Another line of generalization concerns \emph{generalized Poisson brackets}~\cite{deAzcarraga,deAzcarraga2,deAzcarraga3}, with detailed comparisons between these two approaches appearing in \cite{IbanezR3}.
		
		Just as Poisson brackets are associated with bivector fields satisfying the Jacobi identity, Nambu--Poisson brackets are generated by higher-order contravariant skew-symmetric tensors \( N \) of order \( m \) that satisfy a higher-order version of the Jacobi identity, known as the \emph{fundamental identity}. It has been proven~\cite{Alekseevsky,Gautheron,Grabowski} that for \( m \geq 3 \), every such Nambu--Poisson tensor must be decomposable. This implies that any Nambu--Poisson manifold admits a local foliation structure—i.e., it can be decomposed into lower-dimensional leaves where the dynamics effectively unfold.

        On the other hand, Azcárraga et al. \cite{deAzcarraga} proposed an alternative approach to generalizing Poisson brackets, distinct from the traditional Nambu--Poisson framework. Observing that an {{$m$-bracket}} can be induced by a skew-symmetric contravariant tensor of {order $m$}, they derived a new type of fundamental identity from an integrability condition imposed on the tensor. If such a tensor satisfies this identity, it is referred to as a \emph{generalized Poisson tensor}. 
Subsequently, Raul Ib\'a\~nez and collaborators \cite{IbanezR} proved that every even-order Nambu--Poisson bracket is indeed a generalized Poisson bracket. Later, Michor and Vaisman \cite{Michor2} provided criteria for distinguishing generalized Poisson tensors of even and odd orders. According to their characterization, every decomposable multivector field of order greater than or equal to three is a generalized Poisson tensor. Since every Nambu tensor of order $m\geq 3$ is necessarily decomposable, it follows that every Nambu tensor of order $m\geq 3$ is also a generalized Poisson tensor. However, the converse does not hold.

In this work, our interest lies precisely in the difference between these two structures—that is, in those generalized Poisson tensors that are not Nambu tensors. Given the broad and somewhat universal nature of the property that all decomposable multivector fields of order at least three define generalized Poisson tensors, we are particularly interested in the \emph{non-decomposable} ones. Specifically, we aim to investigate under what conditions a dynamical system can be formulated as a (quasi) Hamiltonian system with respect to a non-decomposable generalized Poisson tensor.

	The structure of the paper is as follows. In Section~2, we introduce the notions of Nambu structures and generalized Poisson structures, along with the relevant notation and foundational tools. In Section~3, we present the formation of higher-order Poisson brackets derived from symmetry considerations and conservation laws. Section~4 is devoted to applications of our theoretical framework. Finally, in Section~5, we provide a summary and conclusion.
	%	In summary, this paper extends the classical Liouville–Arnold theory and contributes a unified framework for the Lie-type integrability of Hamiltonian systems beyond the standard symplectic setting. It opens new directions for the study of complex dynamics on geometrically rich manifolds.
	\section{Nambu Structures and Generalized Poisson Structures: Notation and Basic Tools} Let $M$ be a smooth $n$-dimensional manifold and denote by $\Omega^k(M)$ and $\mathfrak{X}^k(M)$ the space of $k$-forms and $k$-vector fields on $M$ respectively, for $k = 0, 1, \dots, n$. We also denote by $\Omega^{\bullet}(M)$ and $\mathfrak{X}^{\bullet}(M)$ the graded spaces of differential forms and multivector fields on $M$. Let  $C^{\infty}(M)$ denote the algebra of smooth real-valued functions on $M$.
	In Hamiltonian mechanics,  manifold $M$ is called a \textit{Poisson manifold} if on its function ring $ A:=C^{\infty}(M)$ there exists a bilinear map
	\[
	\{\,\cdot , \cdot \,\} : A \otimes A \to A
	\]
	satisfying the following properties:
	
	\begin{enumerate}
		\item \textbf{Skew-symmetry}:
		\[
		\{f_1, f_2\} = -\{f_2, f_1\}, \quad \text{for all } f_1, f_2 \in A.
		\]
		
		\item \textbf{Leibniz rule (derivation property)}:
		\[
		\{f_1 f_2, f_3\} = f_1 \{f_2, f_3\} + f_2 \{f_1, f_3\}, \quad \text{for all } f_1, f_2, f_3 \in A.
		\]
		
		\item \textbf{Jacobi identity}:
		\[
		\{f_1, \{f_2, f_3\}\} + \{f_3, \{f_1, f_2\}\} + \{f_2, \{f_3, f_1\}\} = 0, \quad \text{for all } f_1, f_2, f_3 \in A.
		\]	
	\end{enumerate}
			
	The corresponding binary operation $\{\, \cdot, \cdot\,\}$ on $A$ is called the \emph{Poisson bracket}, and it plays a fundamental role in classical mechanics. Namely, according to Hamilton, the dynamics on the phase space $M$ is determined by a distinguished function $H \in A$, called the \emph{Hamiltonian}, and is governed by the Hamiltonian equations of motion:
	\[
	\frac{d}{dt}f = \{f, H\}, \quad f \in A.
	\]
	
	A \emph{Nambu--Poisson structure of order} $m$ on a smooth manifold $M$ is defined by an $m$-vector field, that is, a $C^{\infty}(M)$-skew-symmetric multilinear map
	\[
	 N : \underbrace{\Omega^1(M) \times \cdots \times \Omega^1(M)}_{m \text{ times}} \longrightarrow C^{\infty}(M),
	\]
	which in local coordinates $(x_1, x_2, \dots, x_n)$ can be written as
	\[
	 N = n_{i_1 \cdots i_m}(x)\, \frac{\partial}{\partial x^{i_1}} \wedge \frac{\partial}{\partial x^{i_2}} \wedge \cdots \wedge \frac{\partial}{\partial x^{i_m}}.
	\]
	
    This multivector field defines an $m$-ary bracket on smooth functions:
	\[
	\{f_1, f_2, \dots, f_m\} = N(df_1, df_2, \dots, df_m),
	\]
	for all $f_1, f_2, \dots, f_m \in C^{\infty}(M)$.
	In local coordinates $(x^1, \dots, x^n)$, the Nambu bracket of $m$ functions is given by
	\[
	\{f_1, f_2, \dots, f_m\} =  n_{i_1 i_2 \cdots i_m}\frac{\partial f_1}{\partial x^{i_1}} \frac{\partial f_2}{\partial x^{i_2}} \cdots \frac{\partial f_m}{\partial x^{i_m}},
	\]
	where summation over repeated indices is understood. The bracket $\{f_1, f_2, \dots, f_m\}$ satisfies the following properties:
	
	\begin{enumerate}
		\item \textbf{Skew-symmetry:} For any $f_1, \dots, f_m \in C^{\infty}(M)$ and any permutation $\sigma \in S_m$ (the symmetric group on $m$ elements), we have
		\[
		\{f_1, f_2, \dots, f_m\} = \operatorname{sgn}(\sigma)\, \{f_{\sigma(1)}, f_{\sigma(2)}, \dots, f_{\sigma(m)}\},
		\]
		where $\operatorname{sgn}(\sigma)$ denotes the sign (parity) of the permutation $\sigma$.
		
		\item \textbf{Multilinearity:} For any real numbers $k_1, k_2$ and functions $g_1, g_2, f_2, \dots, f_m \in C^{\infty}(M)$,
		\[
		\{k_1 g_1 + k_2 g_2, f_2, \dots, f_m\} = k_1 \{g_1, f_2, \dots, f_m\} + k_2 \{g_2, f_2, \dots, f_m\}.
		\]
		
		\item \textbf{Leibniz rule:} For any functions $g_1, g_2, f_2, \dots, f_m \in C^{\infty}(M)$,
		\[
		\{g_1 g_2, f_2, \dots, f_m\} = g_1 \{g_2, f_2, \dots, f_m\} + \{g_1, f_2, \dots, f_m\} g_2.
		\]
		
		\item \textbf{Fundamental identity (FI):} Also known as the Takhtajan identity, this property generalizes the Jacobi identity. For any $2m - 1$ functions $f_1, \dots, f_{m-1}, g_m, \dots, g_{2m-1} \in C^{\infty}(M)$,
	\begin{equation}\label{fi1}
		\{f_1, \dots, f_{m-1}, \{g_m, \dots, g_{2m-1}\}\} =
		\sum_{j = m}^{2m - 1} \{g_m, \dots, \{f_1, \dots, f_{m-1}, g_j\}, \dots, g_{2m-1}\},
	\end{equation}
		where in each term, the inner bracket replaces the $j$-th function in the outer bracket.
	\end{enumerate}
	\begin{remark}
		Property 4 in \eqref{fi1}, known as the \emph{Fundamental Identity}, is considered the appropriate generalization of the Jacobi identity characterizing the standard Poisson bracket. 
	\end{remark}

Now, let us consider another type of bracket, which also generalizes the Poisson bracket. If \( M \) is an \( n \)-dimensional manifold, an \( m \)-{ary Poisson bracket} also called a generalized $m$-Poisson structure in \cite{deAzcarraga,deAzcarraga3}), associated with a Poisson \( m \)-tensor \( \mathcal{N} \), is a bracket of the form
\[
\{f_1, \dots, f_m\} =\mathcal{N}(df_1, \dots, df_m) \quad \text{for } f_1, \dots, f_m \in C^\infty(M),
\]
where \( \mathcal{N} \in \mathfrak{X}^m(M) \) is an \( m \)-vector field, and the following generalized Jacobi identity of order \( m \) \cite{Michor} is satisfied:
\begin{align}\label{GPo}
  \sum_{\sigma \in S_{2m-1}} (\operatorname{sign}\sigma)\{\{f_{\sigma_1}, \dots, f_{\sigma_m}\}, f_{\sigma_{m+1}}, \dots, f_{\sigma_{2m-1}}\} = 0,  
\end{align}
where \( S_{2m-1} \) denotes the symmetric group on \( 2m - 1 \) elements.

Furthermore, we recall that the definition of Schouten bracket. As is well known,	the Schouten bracket is the unique bilinear map 
\[
[\cdot, \cdot]_S : \mathfrak{X}^\bullet(M) \times \mathfrak{X}^\bullet(M) \to \mathfrak{X}^\bullet(M)
\]
satisfying the following properties:

\begin{itemize}
	\item If $X \in \mathfrak{X}^k(M)$ and $Y \in \mathfrak{X}^l(M)$, then $[X, Y]_S \in \mathfrak{X}^{k+l-1}(M)$;
	\item $[X, Y]_S = -(-1)^{(k+1)(l+1)}[Y, X]_S$;
	\item For vector fields, it coincides with the standard Lie bracket;
	\item For $X \in \mathfrak{X}^k(M)$, $Y \in \mathfrak{X}^l(M)$, $Z \in \mathfrak{X}^m(M)$,
	\[
	[X, Y \wedge Z]_S = [X, Y]_S \wedge Z + (-1)^{(k-1)l} Y \wedge [X, Z]_S;
	\]
	\item It satisfies the graded Jacobi identity:
	\[
(-1)^{(k-1)(m-1)} [X, [Y, Z]_S ]_S
+ (-1)^{(l-1)(k-1)} [Y, [Z, X]_S ]_S
+ (-1)^{(m-1)(l-1)} [Z, [X, Y]_S ]_S
= 0.
\]
\end{itemize}

Specially, for two decomposable $m$-vector fields $X_1 \wedge \cdots \wedge X_m$ and $Y_1 \wedge \cdots \wedge Y_m$, their  Schouten bracket is given by the following expression:
\begin{align}\label{E2}
	[X_1 \wedge \cdots \wedge X_m,\, Y_1 \wedge \cdots \wedge Y_m]_S =
	\sum_{i=1}^{m} \sum_{j=1}^{m} (-1)^{i+j} [X_i, Y_j] \wedge X_1 \wedge \cdots \wedge \widehat{X}_i \wedge \cdots \wedge X_m \wedge Y_1 \wedge \cdots \wedge \widehat{Y}_j \wedge \cdots \wedge Y_m,
\end{align}
where $\widehat{X}_i$ and $\widehat{Y}_j$ indicate omission of the vector fields $X_i$ and $Y_j$, respectively. 

\medskip

{\bf Note:} {For simplicity, henceforth, we will refer to the bracket $[\cdot,\cdot]_S$ simply by \([\cdot, \cdot]\) to denote the Schouten bracket between multivector fields, which in the case of two vector fields reduces to their Lie bracket.
}

\begin{remark}
The \( m \)-ary Poisson bracket defined above naturally satisfies skew-symmetry, multilinearity, and the Leibniz rule. However, the key difference between the Nambu bracket and the generalized $m$-Poisson bracket lies in the distinct forms of their generalized Jacobi identities, respectively, \eqref{fi1} and \eqref{GPo}
.\end{remark}

\begin{proposition}\label{PG}
The multivector field \(\mathcal{N} \in \mathfrak{X}^m(M) \) on an $n$-dimensional manifold $M$ defines a generalized $m$-Poisson bracket if and only if either:
\begin{itemize}
    \item \( m \) is even and the Schouten–Nijenhuis bracket \([\mathcal{N}, \mathcal{N}] = 0\), or
    \item \( m \) is odd and \( \mathcal{N} \) satisfies the following conditions:
    \begin{description}
        \item[(A)] {(Plücker condition)} $(i(\alpha)\mathcal{N}) \wedge (i(\beta)\mathcal{N}) = 0
        \quad \forall\, \alpha, \beta \in T^*M$,
        
        \item[(B)] $\sum_{u=1}^n (i(dx^u)\mathcal{N}) \wedge (L_{\partial/\partial x^u}\mathcal{N}) = 0,$
     \end{description}
    where \((x^u)\) are local coordinates on \( M \), and \( L \) denotes the Lie derivative \cite{Michor2}.
\end{itemize}
\end{proposition}
\begin{definition}
A multivector field \( {\mathcal{N}} \in \mathfrak{X}^m(M) \) of degree \(m\) on a smooth manifold $n$-dimensional manifold \(M\) is said to be \emph{decomposable} if it can be written as the wedge product of \(m\) vector fields, that is,
\[
{\mathcal{N}} = X_1 \wedge X_2 \wedge \cdots \wedge X_m,
\]
for some vector fields \(X_1, X_2, \dots, X_m \in \mathfrak{X}(M)\). Otherwise, \(\mathcal{N}\) is said to be \emph{non-decomposable}.
\end{definition}

\begin{remark}\label{Impor}
It is easy to verify that any decomposable multivector field of order \( m \), with \( m > 2 \), defines a generalized Poisson structure. Indeed, if \(\mathcal{N} = X_1 \wedge \cdots \wedge X_{2k}\), with \(k \geq 1\), then by standard properties of the Nambu–Poisson bracket, we have
\begin{align}\label{Even}
[\mathcal{N}, \mathcal{N}] &= [X_1 \wedge \cdots \wedge X_{2k},\, X_1 \wedge \cdots \wedge X_{2k}]\nonumber\\
&= \sum_{1\leq i<j\leq 2k } (-1)^{i+j} X_1 \wedge \cdots \wedge \widehat{X}_i \wedge \cdots \wedge X_{2k} \wedge [X_i, X_j] \wedge X_1 \wedge \cdots \wedge \widehat{X}_j \wedge \cdots \wedge X_{2k} = 0,
\end{align}
due to the appearance of repeated vector fields. 

If  \(\mathcal{N} = X_1 \wedge \cdots \wedge X_{2k+1}\), then for any \(\alpha, \beta \in T^*M\), it is clear that
\[
(i(\alpha)\mathcal{N}) \wedge (i(\beta)\mathcal{N}) = 0,
\]
and
\[
\sum_{u=1}^n (i(dx^u)\mathcal{N}) \wedge (L_{\partial/\partial x^u}\mathcal{N}) = 0,
\]
again due to the appearance of repeated vector fields.
\end{remark}
We can now provide the following important guidance, as highlighted in Corollary III.8 of \cite{IbanezR}.
\begin{remark}\label{Nambu}
 Every Nambu-Poisson tensor of even order is a generalized Poisson tensor, but the
 converse does not hold.
\end{remark}
It is shown in \cite{Gautheron} that any Nambu tensor of order $m\geq 3$ is decomposable—a fact that was conjectured in \cite{Takhtajan1994} and later found to be a consequence of an old result by Weitzenböck \cite{Weitzenbock}, reproduced in a textbook by Schouten \cite{Schouten}. Therefore, if an \( m \)-vector is not decomposable, then it is not a Nambu tensor; however, it may still be a generalized Poisson \( m \)-tensor. Together with Remark \ref{Impor}, \ref{Nambu} and discussion we can get the following more generalized proposition.
\begin{proposition}
     Every Nambu-Poisson tensor of order $m\geq 3$  is a generalized Poisson tensor, but the
 converse does not hold.
\end{proposition}

A fundamental feature of a Nambu--Poisson structure or a generalized Poisson $m$-tensor is that any $(m - 1)$ smooth functions $f_1, \dots, f_{m-1} \in C^{\infty}(M)$ determine a vector field, denoted by $X_{f_1, \dots, f_{m-1}}$, via contraction of the $m$-vector field $N$ with the $(m - 1)$-form $df_1 \wedge \cdots \wedge df_{m-1}$. This vector field acts on a smooth function $g$ according to
\[
X_{f_1, \dots, f_{m-1}}(g) = \{f_1, \dots, f_{m-1}, g\}.
\]

Moreover, if \(\mathcal{N} \) is a Nambu tensor $N$, then the vector field \( X_{f_1, \dots, f_{m-1}} \) satisfies the invariance condition
\[
L_{X_{f_1, \dots, f_{m-1}}} N = 0,
\]
i.e., it preserves the Nambu tensor \( N \) under the Lie derivative. Importantly, this invariance condition for all \((m - 1)\)-tuples of functions is equivalent to the Fundamental Identity \eqref{fi1}. 

{\bf Note:} From now on, when a generalized Poisson tensor $\mathcal{N}$ of order $m$ is a Nambu-Poisson tensor, we will refer to it as $N$, instead of caligraphic $\mathcal{N}.$

\medskip
In fact, the condition \( L_{X_{f_1, \dots, f_{m-1}}} N = 0 \) implies that the vector field \( X_{f_1, \dots, f_{m-1}} \) acts as a derivation of the Nambu bracket:
\[
X_{f_1, \dots, f_{m-1}} \big( \{g_1, \dots, g_m\} \big)
= \sum_{j = 1}^{m} \{g_1, \dots, X_{f_1, \dots, f_{m-1}}(g_j), \dots, g_m\}.
\]
However, when the tensor is a generalized Poisson $\mathcal{N}$ of order \( m \), such invariance generally does not hold. Whether \( N \) is an \( m \)-order Nambu tensor or a generalized Poisson tensor $\mathcal{N}$ of order $m$, a smooth function \( f \in C^{\infty}(M) \) is said to be a \emph{constant of motion} for the dynamics generated by the vector field \( X_{f_1, \dots, f_{m-1}} \) if and only if
\[
\{f, f_1, f_2, \dots, f_{m-1}\} = 0.
\]

\vspace{0.5em}

Furthermore, in the case where \( N \) is a Nambu tensor, the Nambu--Poisson bracket of \( m \) constants of motion is itself a constant of motion. That is, the set of first integrals is closed under the Nambu bracket. Indeed, suppose that \( g_1, \dots, g_m \) are constants of motion for the vector field \( X_{f_1, \dots, f_{m-1}} \), then we have
\begin{align*}
L_{X_{f_1, \dots, f_{m-1}}}\{g_1, \dots, g_m\}
&= L_{X_{f_1, \dots, f_{m-1}}}\big(N(dg_1, dg_2, \dots, dg_m)\big) \\
&= (L_{X_{f_1, \dots, f_{m-1}}} N)(dg_1, dg_2, \dots, dg_m) \\
&\quad + \sum_{j = 1}^{m} N(dg_1, \dots, d\,L_{X_{f_1, \dots, f_{m-1}}}g_j, \dots, dg_m)=0.
\end{align*}

However, in the case of a generalized Poisson tensor \( \mathcal{N} \), this closure property generally fails. Since the vector field \( X_{f_1, \dots, f_{m-1}} \) may not preserve the tensor \( \mathcal{N} \), we cannot guarantee that the \( m \)-ary Poisson bracket of \( m \) constants of motion is itself a constant of motion.

A vector field of the form \( X_{f_1, \dots, f_{m-1}} \) is called a \emph{Nambu--Hamiltonian vector field} when \( N \) is a Nambu tensor, and a \emph{generalized Hamiltonian vector field} when \( \mathcal{N} \) is a generalized  \( m \)-Poisson tensor. In both cases, each of the functions \( f_1, \dots, f_{m-1} \) must be a constant of motion. This implies that only those vector fields admitting sufficiently many independent first integrals can arise as Nambu--Hamiltonian or generalized Hamiltonian vector fields with respect to a given Nambu--Poisson structure or generalized \( m \)-Poisson  tensor.

In \cite{Cari}, a more general notion is introduced for the case when \( N \) is a Nambu tensor: a vector field \( Y \) on \( M \) is called a \emph{quasi-Nambu--Hamiltonian vector field} if there exists a smooth function \( g \) such that \( gY \) is Nambu--Hamiltonian. That is, there exist functions \( f_1, \dots, f_{m-1} \) such that
\[
gY = X_{f_1, \dots, f_{m-1}}.
\]
In this case, the functions \( f_i \) must also be constants of motion for the associated dynamics, due to the skew-symmetry of the Nambu bracket. Similarly,  for a generalized \( m \)-Poisson tensor \( \mathcal{N} \), we say that a vector field \( \mathcal{J} \) is a \emph{quasi-generalized Hamiltonian vector field} if there exists a smooth function \( g \) such that \( g\mathcal{J} \) is generalized Hamiltonian, where $g$ is called a multiplier. When $g=1$, $\mathcal{J}$ is called \emph{generalized Hamiltonian vector field}. 

%An especially notable case occurs when the order $m$ of the Nambu structure equals the dimension of the manifold $M$. For instance, if $Q$ is an $n$-dimensional manifold and $M = T^*Q$ is its cotangent bundle equipped with the canonical symplectic form $\omega_0$, then the multi-vector field on $M$ that is dual to the volume form $\omega_0^{\wedge n}$ defines a Nambu--Poisson structure. In this setting, the Nambu structure is identified with the dual of the Liouville volume form.

%An important structural property of Nambu--Poisson geometry is that, when the order $m$ is even, the \emph{Fundamental Identity} (FI) is equivalent to the vanishing of the \emph{Schouten--Nijenhuis bracket} of the corresponding $m$-vector field $N$ with itself:
%\[
%[N, N] = 0,
%\]
%see, e.g., \cite{deAzcarraga, deAzcarraga3}. This condition serves as a natural generalization of the Jacobi identity for bivector fields in Poisson geometry.

	\section{Formation of Higher-Order Poisson Brackets through Symmetry and Conservation}
   According to Remark~\ref{Impor}, we know that every decomposable multivector field of order \( m \), with \( m > 2 \), defines a generalized Poisson structure. Therefore, constructing a decomposable generalized Poisson structure is straightforward, and in this section, we focus on the non-decomposable case.

We consider a dynamical vector field defined on an \( n \)-dimensional manifold \( M \) and investigate the conditions under which it can be regarded as a quasi-generalized Hamiltonian system associated with a non-decomposable generalized \( m \)-Poisson tensor. 

It was shown in \cite{Cari} that if the system admits two independent symmetries, then a Nambu structure of order three can be constructed, making the system quasi-Hamiltonian with respect to that structure.

In this section, we generalize this result to the case where the system is coupled with commuting independent symmetries. We demonstrate that the existence of such a family of symmetries allows the construction of a higher-order non-decomposable generalized $m$-Poisson structure, under which the system becomes quasi-generalized Hamiltonian. This provides a systematic approach to identifying dynamical systems for a broader class of generalized Hamiltonian structures.

\begin{theorem}\label{T1}
    Let $\Gamma$ be a dynamical vector field on an $n$-dimensional smooth manifold $M$. Assume the following conditions hold:
    \begin{enumerate}
        \item Given any $5\leq s \leq n-1$ independent vector fields, represented by vector fields $X_1, \dots, X_s$, satisfying 
        \[
        \Gamma \wedge X_1 \wedge \cdots \wedge X_s \not\equiv 0, \quad [X_i, \Gamma] = 0, \quad [X_i, X_j] = 0, \quad i,j = 1, \dots, s.
        \]
        \item There exist $(s-2)$  functionally independent first integrals (constants of motion) $h_1, \dots, h_{s-2}$ of $\Gamma$, i.e.,
        \[
        \Gamma(h_1) = \cdots = \Gamma(h_{s-2}) = 0.
        \]
        \item The contraction of the $(s-2)$-vector field $X_1 \wedge \cdots \wedge  X_{s-4}\wedge X_{s-3}\wedge X_{s-2} + X_1 \wedge \cdots \wedge X_{s-4} \wedge X_{s-1} \wedge X_s$ with the $(s-2)$-form $dh_1 \wedge \cdots \wedge dh_{s-2}$ is non-vanishing on an open dense subset of $M$.
    \end{enumerate}
    Then the multi-vector field
    \[
    \mathcal{N} = \Gamma \wedge X_1 \wedge \cdots\wedge X_{s-4}\wedge X_{s-3} \wedge X_{s-2} + \Gamma \wedge X_1 \wedge \cdots \wedge X_{s-4} \wedge X_{s-1} \wedge X_s
    \]
    defines a non-decomposable generalized $(s-1)$-Poisson tensor on $M$, and the system is \emph{quasi-generalized Hamiltonian} with respect to $\mathcal{N}$. Moreover, the multiplier is a first integral of the system.

    Further, when $2s>n+4,$ there exists a generalized $(s-1)$-Poisson structure $\mathcal{J}$ on $M$, proportional to $\mathcal{N}$, such that $\Gamma$ is the geneneralized Hamiltonian vector field associated with the multi-pair of functions $(h_1,...,h_{s-2})$ under the generalized Hamiltonian dynamics defined by $\mathcal{J}$.
\end{theorem}
\begin{proof}
    When \(s\) is odd, the multivector field
    \[
    \mathcal{N} = \Gamma \wedge X_1 \wedge \cdots\wedge X_{s-4}\wedge X_{s-3} \wedge X_{s-2} + \Gamma \wedge X_1 \wedge \cdots \wedge X_{s-4} \wedge X_{s-1} \wedge X_s
    \]
    defines an even \((s-1)\)-vector field on the manifold \(M\).

    Since \(X_1, \dots, X_s\) are commuting infinitesimal symmetries of the dynamical vector field \(\Gamma\), i.e.,
    \[
    [X_i, \Gamma] = 0, \quad [X_i, X_j] = 0, \quad i,j = 1, \dots, s,
    \]
    it follows that the distribution generated by \(\{\Gamma, X_1, \dots, X_s\}\) is integrable.

    According to Proposition \ref{PG}, such an integrable distribution implies that the multivector field \(\mathcal{N}\) satisfies the \emph{generalized Jacobi identity of order \((s-1)\)} \eqref{GPo}. Hence, \(\mathcal{N}\) defines an \((s-1)\)-Poisson tensor on \(M\). 

    Since \(s\) is odd, \(s+1\) is even, so the tensor \(\mathcal{N}\) is an even-order multivector field. It is well known that the \emph{generalized Jacobi identity} of an even-order multivector field \(\mathcal{N}\) is equivalent to the vanishing of the \emph{Schouten--Nijenhuis bracket} of \(\mathcal{N}\) with itself:
    \[
    [\mathcal{N}, \mathcal{N}] = 0.
    \]
    In our case, the Schouten--Nijenhuis bracket of \(\mathcal{N}\) with itself vanishes. For convenience for the following calculation, let us denote \(X_{s+1} = \Gamma\). Then, by equations \eqref{E2} and \eqref{Even}, we compute:
    \begin{align*}
    [\mathcal{N}, \mathcal{N}] &= [X_1 \wedge \cdots \wedge X_{s-4}\wedge X_{s-3} \wedge X_{s-2} \wedge X_{s+1} + X_1 \wedge \cdots \wedge X_{s-4} \wedge X_{s-1} \wedge X_s \wedge X_{s+1}, \\
    &\quad X_1 \wedge \cdots \wedge X_{s-4}\wedge X_{s-3} \wedge X_{s-2} \wedge X_{s+1} + X_1 \wedge \cdots \wedge X_{s-4} \wedge X_{s-1} \wedge X_s \wedge X_{s+1}] \\
    &= [X_1 \wedge \cdots \wedge X_{s-4}\wedge X_{s-3} \wedge X_{s-2} \wedge X_{s+1}, X_1 \wedge \cdots \wedge X_{s-4}\wedge X_{s-3} \wedge X_{s-2} \wedge X_{s+1}] \\
    &+ [X_1 \wedge \cdots \wedge X_{s-4} \wedge X_{s-1} \wedge X_s \wedge X_{s+1}, X_1 \wedge \cdots \wedge X_{s-4} \wedge X_{s-1} \wedge X_s \wedge X_{s+1}] \\
    &+ 2 [X_1 \wedge \cdots \wedge X_{s-4}\wedge X_{s-3} \wedge X_{s-2} \wedge X_{s+1}, X_1 \wedge \cdots \wedge X_{s-4} \wedge X_{s-1} \wedge X_s \wedge X_{s+1}] \\
    &= 2 [X_1 \wedge \cdots \wedge X_{s-4}\wedge X_{s-3} \wedge X_{s-2} \wedge X_{s+1}, X_1 \wedge \cdots \wedge X_{s-4} \wedge X_{s-1} \wedge X_s \wedge X_{s+1}] \\
    &= 0.
    \end{align*}
    This confirms that \(\mathcal{N}\) defines a non-decomposable \((s-1)\)-Poisson tensor.
    Furthermore, remember that \(\Gamma = X_{s+1}\), since
    \begin{align*}
    L_{X_{s+1}} (X_1 \wedge \cdots\wedge X_{s-4}\wedge X_{s-3} \wedge X_{s-2} \wedge X_{s+1} + X_1 \wedge \cdots \wedge X_{s-4} \wedge X_{s-1} \wedge X_s \wedge X_{s+1}) = 0,
    \end{align*}
    we conclude that \(\mathcal{N}\) is invariant under the flow of \(\Gamma\), i.e., \(L_\Gamma \mathcal{N} = 0\).

    Let \(\{x_i\}_{i=1}^{n}\) be a local coordinate system on \(M\), and suppose the vector fields are expressed as
    \[
    \Gamma = f_{j}(x)\frac{\partial}{\partial x^j}, \quad X_i = {g_{ij}}(x)\frac{\partial}{\partial x^j}, \quad i=1,...,s,j = 1, \dots, n.
    \]
    Then \(\mathcal{N}\) takes the local form
    {
    \[
    \mathcal{N} = f_{j_0} g_{1j_1} \cdots g_{s-4j_{s-4}} \frac{\partial}{\partial x^{j_0}} \wedge \cdots \wedge \frac{\partial}{\partial x^{j_{s-4}}}\wedge \left(g_{s-3j_{s-3}} g_{s-2j_{s-2}}\frac{\partial}{\partial x^{j_{s-3}}}\wedge \frac{\partial}{\partial x^{j_{s-2}}}+g_{s-1j_{s-1}} g_{sj_{s}}\frac{\partial}{\partial x^{j_{s-1}}}\wedge \frac{\partial}{\partial x^{j_{s}}}\right).
    \]\\
   }

    Let \(h_1, \dots, h_{s-2}\) be constants of motion for the vector field \(\Gamma\). The action of the Nambu tensor \(\mathcal{N}\) on the differentials \(dh_1, \dots, dh_{s-2}\) gives:
    \begin{align}\label{MathcalN}
        \mathcal{N}(\cdot, dh_1, \dots, dh_{s-2}) = \tilde{h} \Gamma,
    \end{align}
    where
    \[
    \tilde{h} = \det
    \begin{pmatrix}
    X_1(h_1) & \cdots & X_1(h_{s-2}) \\
    \vdots & \ddots & \vdots \\
    X_{s-4}(h_1) & \cdots & X_{s-4}(h_{s-2})\\
    X_{s-3}(h_1) & \cdots & X_{s-3}(h_{s-2})\\
    X_{s-2}(h_1) & \cdots & X_{s-2}(h_{s-2})
    \end{pmatrix}
    + \det
    \begin{pmatrix}
    X_1(h_1) & \cdots & X_1(h_{s-2}) \\
    \vdots & \ddots & \vdots \\
    X_{s-4}(h_1) & \cdots & X_{s-4}(h_{s-2}) \\
    X_{s-1}(h_1) & \cdots & X_{s-1}(h_{s-2}) \\
    X_s(h_1) & \cdots & X_s(h_{s-2})
    \end{pmatrix}.
    \]
    This shows that the vector field \(\Gamma\) is \emph{quasi-generalized Hamiltonian} with respect to the \((s-1)\)-Poisson tensor \(\mathcal{N}\).

    Moreover, since the vector fields \(X_1, \dots, X_s\) commute with the dynamical vector field \(\Gamma\), and the functions \(h_1, \dots, h_s\) are constants of motion for \(\Gamma\), we have
    \[
    {L}_\Gamma \left(X_i(h_{\sigma(i)})\right) = [\Gamma, X_i](h_{\sigma(i)}) + X_i\left(\Gamma(h_{\sigma(i)})\right) = 0,
    \]
    for each \(i = 1, \dots, s\).
    Therefore, by the Leibniz rule for the Lie derivative, we conclude that the function \(\tilde{h}\) is also a constant of motion for \(\Gamma\):
    \[
    L_\Gamma \tilde{h} = 0.
    \]

    When \(s\) is even, again according to Proposition \ref{PG}, we can verify that
    \[
    (i(\alpha)\mathcal{N}) \wedge (i(\beta)\mathcal{N}) = 0 \quad \forall\, \alpha, \beta \in T^*M,
    \]
    \[
    \sum_{u=1}^n (i(dx^u)\mathcal{N}) \wedge (L_{\partial/\partial x^u}\mathcal{N}) = 0,
    \]
    hence, \(\mathcal{N}\) is an odd \((s-1)\)-vector field, which is a generalized \((s-1)\)-Poisson tensor. Similar to the discussion above, we also conclude that \(\Gamma\) is a quasi-Hamiltonian vector field with respect to \(\mathcal{N}\), with the multiplier
    \[
    \tilde{h} = \det
    \begin{pmatrix}
    X_1(h_1) & \cdots & X_1(h_{s-2}) \\
    \vdots & \ddots & \vdots \\
    X_{s-4}(h_1) & \cdots & X_{s-4}(h_{s-2})\\
    X_{s-3}(h_1) & \cdots & X_{s-3}(h_{s-2})\\
    X_{s-2}(h_1) & \cdots & X_{s-2}(h_{s-2})
    \end{pmatrix}
    + \det
    \begin{pmatrix}
    X_1(h_1) & \cdots & X_1(h_{s-2}) \\
    \vdots & \ddots & \vdots \\
    X_{s-4}(h_1) & \cdots & X_{s-4}(h_{s-2}) \\
    X_{s-1}(h_1) & \cdots & X_{s-1}(h_{s-2}) \\
    X_s(h_1) & \cdots & X_s(h_{s-2})
    \end{pmatrix}.
    \]
When $2s > n + 4$, let $\mathcal{J} = \frac{1}{\tilde{h}} \mathcal{N}$. It can be verified that, regardless of whether $s - 1$ is even or odd, $\mathcal{J}$ satisfies the conditions of Proposition~\ref{PG}. Hence, $\mathcal{J}$ is also a generalized $(s{-}1)$-Poisson tensor, differing from $\mathcal{N}$ by a scalar multiple $\frac{1}{\tilde{h}}$. According to equation~\eqref{MathcalN}, we conclude that the original system $\Gamma$ is a generalized Hamiltonian system with respect to the generalized Poisson tensor $\mathcal{J}$ associated with the multi-pair of functions $(h_1, \dots, h_{s-2})$.
This completes the proof. 
\end{proof}

	\begin{remark}
	    It is worth noticing that condition 1 of Theorem \ref{T1} can be relaxed: instead of requiring the existence of $m$ commuting symmetries, one may assume the existence of $s$ vector field such that
        $$[X_i,X_l]=0,\quad [X_l,X_k]=0,$$
$$[X_i,X_j]=a_1X_1+\cdots+a_{i-1}X_{i-1}+a_{i+1}X_{i+1}+\cdots +a_{j-1}X_{j-1}+a_{j+1}X_{j+1}+\cdots+a_{s-4}X_{s-4},$$
        where $a_{i}\in C^\infty(M),i,j=1,...,s-4,{i\neq j,l,k=s-3,\cdots,s,}$
        then the proof of Theorem \ref{T1} remains valid. % In fact, a decomposable $m$-vector field $V$ tangent to the distribution, satisfying $X \wedge V \neq 0$ everywhere on $M$, defines a Nambu structure of the form
		%\[
		%\mathcal N = X \wedge V.
		%\]
       
	\end{remark}

     {{\begin{remark}\label{RPla}
     In our construction method, in order for $\mathcal{N}$ to be a non-decomposable generalized Poisson tensor, $\mathcal{N}$ must be at least of order $4$, which implies that we need at least $5$ commuting symmetries, as required by the method described in Theorem~\ref{T1}. 

To be more precise, if we aim to construct a generalized 4-Poisson tensor, it must take a form such as
\[
\mathcal{N} = \Gamma \wedge X_1 \wedge X_2 \wedge X_3 + \Gamma \wedge X_1 \wedge X_4 \wedge X_5,
\]
which guarantees its non-decomposability. However, this construction already requires at least five independent symmetries $\{X_1, \dots, X_5\}$.

\medskip

One may ask whether it is possible to construct a non-decomposable generalized Poisson tensor of lower degree, such as degree $2$ or $3$. Let us examine both cases.

\paragraph{Degree 3} According to the Plücker condition in Proposition~\ref{PG}, an odd generalized Poisson tensor $\mathcal{N} \in \wedge^3 TM$ must satisfy:
\[
(i(\alpha)\mathcal{N})\wedge (i(\beta)\mathcal{N}) = 0 \qquad \forall\, \alpha, \beta \in T^*M.
\]
However, there no examples of non-decomposable trivector fields that satisfy this condition. To see this property clearer, we consider the manifold $M = \mathbb{R}^6$ with coordinates $(x_1, \dots, x_6)$, and define:
\[
\mathcal{N} = \frac{\partial}{\partial x_1} \wedge \frac{\partial}{\partial x_2} \wedge \frac{\partial}{\partial x_3}
\;+\;
\frac{\partial}{\partial x_4} \wedge \frac{\partial}{\partial x_5} \wedge \frac{\partial}{\partial x_6}.
\]
Clearly, $\mathcal{N}$ is not decomposable, as it is the sum of two decomposable trivectors supported on disjoint coordinate blocks. To check whether $\mathcal{N}$ satisfies the Plücker condition, we compute the following contractions.

\medskip

For $j=1,2,3$:
\[
i(dx^1)\mathcal{N} = \partial_{x_2} \wedge \partial_{x_3}, \quad
i(dx^2)\mathcal{N} = -\partial_{x_1} \wedge \partial_{x_3}, \quad
i(dx^3)\mathcal{N} = \partial_{x_1} \wedge \partial_{x_2}.
\]

For $j=4,5,6$:
\[
i(dx^4)\mathcal{N} = \partial_{x_5} \wedge \partial_{x_6}, \quad
i(dx^5)\mathcal{N}= -\partial_{x_4} \wedge \partial_{x_6}, \quad
i(dx^6)\mathcal{N} = \partial_{x_4} \wedge \partial_{x_5}.
\]

Thus, the set $\{ i(dx^j)\mathcal{N} \}$ consists of two independent families of bivectors:
one supported on $\{x_1, x_2, x_3\}$ and the other on $\{x_4, x_5, x_6\}$. We now analyze the wedge products.

\medskip

- If both contractions come from the same block (e.g., $i(dx^1)\mathcal{N}$ and $i(dx^2)\mathcal{N}$), then their wedge product vanishes because they lie in a $3$-dimensional subspace, and the wedge of two bivectors in $\wedge^2(\mathbb{R}^3)$ is zero.

\medskip

- If the contractions come from different blocks, their wedge product does not vanish. For instance:
\[
(i(dx^1)\mathcal{N})\wedge(i(dx^4)\mathcal{N}) 
= (\partial_{x_2} \wedge \partial_{x_3}) \wedge (\partial_{x_5} \wedge \partial_{x_6})
= \partial_{x_2} \wedge \partial_{x_3} \wedge \partial_{x_5} \wedge \partial_{x_6} \neq 0.
\]
This shows that $\mathcal{N}$ violates the Plücker condition, and hence is not a generalized Poisson tensor.

\paragraph{Degree 2} For degree $2$, non-decomposable bivector fields do exist; however, they cannot be constructed using the method presented in Theorem~\ref{T1}. 

To see why, suppose we only have two commuting vector fields $X_1$ and $X_2$, and construct:
\[
\mathcal{N} = \Gamma \wedge X_1 + \Gamma \wedge X_2 = \Gamma \wedge (X_1 + X_2),
\]
which is manifestly decomposable as a wedge product of two vector fields. Therefore, our construction method only yields decomposable bivector fields in this case.

We will present an alternative approach to construct non-decomposable generalized Poisson tensors of degree $2$ in the following.
     \end{remark}}}
  \begin{theorem}\label{T2}
        Let $\Gamma$ be a dynamical vector field on an $n$-dimensional smooth manifold $M$. Assume the following conditions hold:
    \begin{enumerate}
        \item Given any 3  independent vector fields, represented by vector fields $X_1, \dots, X_3$, satisfying
        \[
        \Gamma \wedge X_1 \wedge X_2\wedge X_3 \not\equiv 0, \quad [X_i, \Gamma] = 0, \quad [X_i, X_j] = 0, \quad i,j = 1, 2,3.
        \]
        \item There exists a functionally independent first integral $h$ of $\Gamma$, such that 
        \[
        \Gamma(h) = X_1(h)=X_2 (h)= 0,\quad X_3(h)=0.
        \]
    \end{enumerate}
    Then the 2-vector field
    \[
    \mathcal{N} = \Gamma \wedge X_1 + \Gamma \wedge X_3+X_1\wedge X_2
    \]
    defines a non-decomposable 2-Poisson tensor on $M$, and the system is \emph{quasi-generalized Hamiltonian} with respect to $\mathcal{N}$. Moreover, the multiplier is a first integral of the system.
  \end{theorem}
\begin{proof}
We now compute the Schouten bracket of the multi-vector field 
\[
\mathcal{N} = \Gamma \wedge X_1 + \Gamma \wedge X_3 + X_1 \wedge X_2
\]
with itself. Given the properties \([X_i, \Gamma] = 0\) and \([X_i, X_j] = 0\) for \(i,j = 1, \dots, 3\), we compute each term of the Schouten bracket.

We can calculate that 
\begin{align*}
    [\mathcal N,\mathcal N]&=[ \Gamma \wedge X_1 + \Gamma \wedge X_3 + X_1 \wedge X_2, \Gamma \wedge X_1 + \Gamma \wedge X_3 + X_1 \wedge X_2]\\
    &=[\Gamma\wedge X_1,\Gamma\wedge X_1]+[\Gamma\wedge X_1,\Gamma\wedge X_3]+[\Gamma\wedge X_1,X_1\wedge X_2]\\
    &+[\Gamma\wedge X_3,\Gamma\wedge X_1]+[\Gamma\wedge X_3,\Gamma\wedge X_3]+[\Gamma\wedge X_3,X_1\wedge X_2]\\
    &+[X_1\wedge X_2,\Gamma\wedge X_1]+[X_1\wedge X_2,\Gamma\wedge X_3]+[X_1\wedge X_2,X_1\wedge X_2]\\
     &=[\Gamma\wedge X_1,\Gamma\wedge X_1]+[\Gamma\wedge X_3,\Gamma\wedge X_3]+[X_1\wedge X_2,X_1\wedge X_2]\\
    &+2[\Gamma\wedge X_3,\Gamma\wedge X_1]+2[X_1\wedge X_2,\Gamma\wedge X_1]+2[X_1\wedge X_2,\Gamma\wedge X_3]
\end{align*}
For $[\Gamma\wedge X_k,\Gamma\wedge X_l],k,l=1,2,$ due to $[X_1,\Gamma]=[X_2,\Gamma]=0,$ we know that 
\begin{align*}
[\Gamma\wedge X_k,\Gamma\wedge X_l]&=[\Gamma\wedge X_k,\Gamma]\wedge X_l-\Gamma\wedge[\Gamma\wedge X_k,X_l]\\
&=-[\Gamma,\Gamma\wedge X_k]\wedge X_l+\Gamma\wedge[X_l,\Gamma\wedge X_k]\\
&=-[\Gamma,\Gamma]\wedge X_k\wedge X_l-\Gamma\wedge [\Gamma,X_k]\wedge X_l+\Gamma\wedge [X_l,\Gamma]\wedge X_k+\Gamma\wedge\Gamma\wedge[X_l,X_k]\\
&=-\Gamma\wedge [\Gamma,X_k]\wedge X_l+\Gamma\wedge [X_l,\Gamma]\wedge X_k=0.
\end{align*}
For $[X_1\wedge X_2,\Gamma\wedge X_j],j=1,3,$ due to $[X_1,\Gamma]=[X_2,\Gamma]=[X_1,X_3]=[X_2,X_3]=0,$ we can get that
\begin{align*}
    [X_1\wedge X_2,\Gamma\wedge X_j]&=[X_1\wedge X_2,\Gamma]\wedge X_j-\Gamma\wedge[X_1\wedge X_2,X_j]\\
    &=-[\Gamma,X_1\wedge X_2]\wedge X_j+\Gamma\wedge[X_j,X_1\wedge X_2]\\
    &=-[\Gamma,X_1]\wedge X_2\wedge X_j-X_1\wedge[\Gamma,X_2]\wedge X_j+\Gamma\wedge[X_j,X_1]\wedge X_2+\Gamma\wedge X_1\wedge[X_j,X_2]=0.
\end{align*}
For $[X_1\wedge X_2,X_1\wedge X_2],$ due to $[X_1,X_2]=0,$ we have
\begin{align*}
    [X_1\wedge X_2,X_1\wedge X_2]&=[X_1\wedge X_2,X_1]\wedge X_2-X_1\wedge[X_1\wedge X_2,X_2]\\
    &=-[X_1,X_1\wedge X_2]\wedge X_2+X_1\wedge[X_2,X_1\wedge X_2]\\
    &=-[X_1,X_1]\wedge X_2\wedge X_2-X_1\wedge[X_1,X_2]\wedge X_2+X_1\wedge[X_2,X_1]\wedge X_2+X_1\wedge X_1\wedge[X_2,X_2]\\
    &=2X_1\wedge[X_2,X_1]\wedge X_2=0.
\end{align*}

Combining all the terms, we conclude:
\[
[\mathcal{N}, \mathcal{N}] = 0.
\]
Thus, \(\mathcal{N}\) is a generalized 2-Poisson tensor. Moreover, we can see that 
$$\mathcal N(\cdot,dh)=X_3(h)\Gamma.$$
 This shows that the vector field $\Gamma$ is \emph{quasi-generalized Hamiltonian} with respect to the 2-Poisson tensor $\mathcal{N}$. Moreover, we calculate that
 $$L_\Gamma(X_3(h))=[\Gamma,X_3](h)+X_3(\Gamma(h))=0.$$ Hence, we complete the proof.
\end{proof}
\begin{remark}\label{R4}
According to Remark~\ref{RPla}, we know that for the case with 4 symmetries, it is not possible to construct a non-decomposable generalized Poisson 3-tensor. However, using the construction method from Theorem~\ref{T2}, one can construct several non-decomposable generalized 2-Poisson tensors. We omit the details here.
\end{remark}
\begin{remark}
It is shown in \cite{Gautheron} that every Nambu tensor is decomposable. However, the generalized Poisson tensors constructed in the two theorems above are non-decomposable. Therefore, they cannot be Nambu tensors. This result shows that for certain dynamical systems, given some generalized Poisson structure, they can be viewed as quasi-generalized Hamiltonian systems but not as quasi-Nambu Hamiltonian systems.
\end{remark}
   
	Through observation, we obtain the following result, which is even more remarkable.
	\begin{theorem}\label{T3}
	Let $\Gamma$ be a vector field satisfying the conditions in Theorem \ref{T1}. Then the multi-vector field 
	\[
	\mathcal{N}_i = \Gamma \wedge X_1 \wedge \cdots \wedge\hat X_i\wedge\cdots\wedge X_{s-4}\wedge X_{s-3}\wedge X_{s-2} + \Gamma \wedge X_1 \wedge \cdots\wedge \hat X_i    \wedge X_{s-4} \wedge X_{s-1} \wedge X_s \quad i=1,...,s-4
	\]
  \noindent  
	defines  $(s-4)$  generalized $(s-2)$-Poisson tensors  on $M$, and the system is \emph{multi-quasi-Hamiltonian} with respect to $\mathcal{N}_i,i=1,...,s-4$. The index is \((s - 4)\) because (except for $\Gamma$) the first \(s - 4\) vector fields in the two parts of  \(N\) are identical. We want the last two vector fields of $\mathcal{N}_i$ to remain unchanged in order to preserve the non-decomposability of \( \mathcal{N}_i\). Therefore, by successively removing the common elements \(X_1, \dots, X_{s-4}\) from both parts of \(\mathcal{N}\), we obtain exactly \( (s - 4) \) such multivector fields \(\mathcal{N}_i\).

	Moreover, there exists generalized Poisson structure $\mathcal{J}_i$ of order $(s-2)$ on $M$, proportional to $\mathcal{N}_i$, such that $\Gamma$ is the Hamiltonian vector field associated with the multi-pair of functions $(h_1,...,\hat h_i,...,h_{s-2})$ under the generalized Poisson structure defined by $\mathcal{J}$,  with $i=1,\dots,s-2$.
	\end{theorem}
\begin{proof}
	Similar to the proof of Theorem \ref{T1}.
\end{proof}

\section{Applications}
\begin{example}
Consider the dynamical vector field
\[
\Gamma = \sum_{i=1}^6 \left( x_{2i} \frac{\partial}{\partial x_{2i-1}} - x_{2i-1} \frac{\partial}{\partial x_{2i}} \right)
\]
on $\mathbb{R}^{12}$ with coordinates $(x_1, \dots, x_{12})$. This vector field generates simultaneous rotations in the $(x_1, x_2)$, $(x_3, x_4)$, $(x_5, x_6)$, $(x_7, x_8)$, $(x_9, x_{10})$, and $(x_{11}, x_{12})$-planes. Physically, one may regard $\Gamma$ as the infinitesimal generator of a system with twelve independent angular momenta evolving in synchrony.

\medskip

Three conserved quantities (first integrals) are
\[
h_1 = x_1 x_3 + x_2 x_4, \quad h_2 = x_1 x_4 - x_2 x_3, \quad h_3 = x_5 x_7 + x_6 x_8.
\]
These invariants satisfy
\[
\Gamma(h_1) = 0, \quad \Gamma(h_2) = 0, \quad \Gamma(h_3) = 0,
\]
so they are constant along the flow of $\Gamma$.

\medskip

We also introduce the ``cross-plane'' rotation vector fields:
\[
Y_1 = x_1 \frac{\partial}{\partial x_3} + x_2 \frac{\partial}{\partial x_4}, \quad
Y_2 = x_1 \frac{\partial}{\partial x_4} - x_2 \frac{\partial}{\partial x_3}, \quad
Y_3 = x_5 \frac{\partial}{\partial x_7} + x_6 \frac{\partial}{\partial x_8},
\]
\[
Y_4 = x_5 \frac{\partial}{\partial x_8} - x_6 \frac{\partial}{\partial x_7}, \quad
Y_5 = x_9 \frac{\partial}{\partial x_{11}} + x_{10} \frac{\partial}{\partial x_{12}}.
\]
These vector fields satisfy the commutation relations
\[
[Y_i, Y_j] = 0, \quad [Y_i, \Gamma] = 0 \quad \text{for } i, j = 1, \dots, 5,
\]
so they are commuting symmetries of the dynamics.

\medskip
The non-degeneracy condition of $ Y_1 \wedge Y_2 \wedge Y_3 + Y_1 \wedge Y_4 \wedge Y_5$ in Theorem~\ref{T1} reduces to computing the matrix
\[
\mathcal{M} = \det \begin{pmatrix}
Y_1(h_1) & Y_1(h_2) & Y_1(h_3) \\
Y_2(h_1) + Y_4(h_1) & Y_2(h_2) + Y_4(h_2) & Y_2(h_3) + Y_4(h_3) \\
Y_3(h_1) + Y_5(h_1) & Y_3(h_2) + Y_5(h_2) & Y_3(h_3) + Y_5(h_3)
\end{pmatrix}
= \begin{pmatrix}
x_1^2 + x_2^2 & 0 & 0 \\
0 & x_1^2 + x_2^2 & 0 \\
0 & 0 & x_5^2 + x_6^2
\end{pmatrix},
\]
with determinant
\[
\det(\mathcal{M}) = (x_1^2 + x_2^2)^2 (x_5^2 + x_6^2).
\]
This determinant is non-zero on the open dense set $\{ (x_1, \dots, x_{12}) \neq (0, \dots, 0) \}$.

\medskip
\paragraph{Generalized 4-Poisson Tensor Construction}
According to Theorem~\ref{T1}, one can form the 4-vector field
\[
\mathcal{N} = \Gamma \wedge Y_1 \wedge Y_2 \wedge Y_3 + \Gamma \wedge Y_1 \wedge Y_4 \wedge Y_5.
\]
Explicitly,
\[
\begin{aligned}
\mathcal{N} &= \sum_{i=1}^6 \left( x_{2i} \frac{\partial}{\partial x_{2i-1}} - x_{2i-1} \frac{\partial}{\partial x_{2i}} \right) \wedge \left( x_1 \frac{\partial}{\partial x_3} + x_2 \frac{\partial}{\partial x_4} \right) \\
&\quad \wedge \left( \left( x_1 \frac{\partial}{\partial x_4} - x_2 \frac{\partial}{\partial x_3} \right) \wedge \left( x_5 \frac{\partial}{\partial x_7} + x_6 \frac{\partial}{\partial x_8} \right) \right. \\
&\quad \left. + \left( x_5 \frac{\partial}{\partial x_8} - x_6 \frac{\partial}{\partial x_7} \right) \wedge \left( x_9 \frac{\partial}{\partial x_{11}} + x_{10} \frac{\partial}{\partial x_{12}} \right) \right).
\end{aligned}
\]
After expansion, $\mathcal{N}$ is a skew-symmetric 4-vector field whose coefficients are polynomials in $(x_1, \dots, x_{12})$.

%\[
%N = (x_1^2 + x_2^2) \Big( x_2 \partial x_1 \wedge \partial x_3 \wedge \partial x_4 - x_1 \partial x_2 \wedge \partial x_3 \wedge \partial x_4 \Big).
%\]

By the assumption on the vector fields $\Gamma, Y_1, \dots, Y_5$, we can deduce that 
\[
[\mathcal{N}, \mathcal{N}] = 0,
\]
hence $\mathcal{N}$ is a generalized 4-Poisson tensor. Moreover, the dynamics $\Gamma$ can be recovered in the quasi-generalized Hamiltonian form
\[
\det(\mathcal{M}) \Gamma(f) = \mathcal{N} \left( dh_1, dh_2, dh_3, df \right), \quad \forall f \in C^\infty(\mathbb{R}^{12}).
\]

%In this case, it is easy to check that $[N,N]=0$ is identically zero. In $\mathbb{R}^4$, a trivector $N$ has degree $3$. 
%The Schouten--Nijenhuis bracket $[N,N]$ would be a $5$-vector. 
%But there are no nonzero $5$-vectors in $4$ dimensions 
%(since they would require $5$ independent directions, but only $4$ exist). 
%Therefore, $[N,N]=0$ holds automatically by dimension.

\end{example}

\begin{example}
Let $M=\mathbb{R}^7$ with coordinates $(x_1,...,x_7)$.

\begin{itemize}
  \item {Dynamical vector field:}
  \[
    \Gamma = \frac{\partial}{\partial x_7}.
  \]

  \item {Commuting symmetries ($m=6$):}
  \[
    X_i = \frac{\partial}{\partial x_i}, \quad
   i=1,...,6.
  \]
  These satisfy
  \[
    [X_i,X_j]=0, \qquad [X_i,\Gamma]=0, \quad i,j=1,...,6.
  \]

  \item {First integrals of $\Gamma$:}
  \[
    h_1=x_1,\quad h_2=x_2,\quad h_3=x_3,\quad h_4=x_4
  \]
  since $\Gamma(h_j)=\tfrac{\partial h_j}{\partial x_7}=0$.

  \item {Nondegeneracy:}  
  The matrix
  \[
    \mathcal M = (X_1\wedge X_2\wedge X_3\wedge X_4+X_1\wedge X_2\wedge X_5\wedge X_6)(dh_1,dh_2,dh_3,dh_4)
      = \begin{pmatrix}
          1 & 0 & 0&0 \\
          0 & 1 & 0&0 \\
          0 & 0 & 1&0\\
          0&0&0&1
        \end{pmatrix}
  \]
  has determinant $\det(\mathcal M)=1\neq 0$.
\end{itemize}

Therefore, by Theorem 3.1, the $5$-vector field
\[
 \mathcal N = \Gamma \wedge X_1 \wedge X_2 \wedge X_3\wedge X_4+\Gamma \wedge X_1\wedge X_2\wedge X_5\wedge X_6
    =  \frac{\partial}{\partial x_7}
      \wedge \frac{\partial}{\partial x_1}
      \wedge \frac{\partial}{\partial x_2}
      \wedge \left(\frac{\partial}{\partial x_3}\wedge \frac{\partial}{\partial x_4}+\frac{\partial}{\partial x_5}
      \wedge \frac{\partial}{\partial x_6}
      \right)
\]
defines a non-decomposable generalized 5-Poisson tensor  on $\mathbb{R}^7$.

Finally, for any smooth function $f \in C^\infty(\mathbb{R}^7)$,
\[
  \Gamma(f) \;=\; \mathcal{N} \big(df,dh_1,dh_2,dh_3,dh_4\big),
\]
so $\Gamma$ is the Hamiltonian vector field associated with the functions
$(h_1,h_2,h_3,h_4)$ under the non-decomposable generalized Poisson structure $\mathcal N$.

\medskip

\noindent\paragraph{Rescaling to $\mathcal J$} In Theorem~\ref{T3} when $2m>n+4,$ one may rescale $\mathcal N$ by the (nonvanishing) factor
$\tilde h=\det\big(X_\ell(h_s)\big)_{ s\neq i}$ so that $\mathcal J=\tilde h^{-1}\mathcal N$ makes $\Gamma$ generalized
Hamiltonian with respect to $(h_1,h_2,h_3,h_4)$. Here $2\cdot 6=12>7+4$ and
\[
\big(X_\ell(h_s)\big)_{ s\neq i} = \mathrm{Id}_4\quad\Longrightarrow\quad \tilde h=1,
\]
hence $\mathcal J=\mathcal N$.

\medskip

Theorem~\ref{T3} associates the $4$-vector fields:

\[
\mathcal{N}_1=\frac{\partial}{\partial x_6}\wedge\frac{\partial}{\partial x_2}\wedge\frac{\partial}{\partial x_3}\wedge \frac{\partial}{\partial x_4}+\frac{\partial}{\partial x_6}\wedge\frac{\partial}{\partial x_2}\wedge\frac{\partial}{\partial x_5}\wedge \frac{\partial}{\partial x_6},\qquad
\mathcal{N}_2=\frac{\partial}{\partial x_6}\wedge\frac{\partial}{\partial x_1}\wedge\frac{\partial}{\partial x_3}\wedge \frac{\partial}{\partial x_4}+\frac{\partial}{\partial x_6}\wedge\frac{\partial}{\partial x_1}\wedge\frac{\partial}{\partial x_5}\wedge \frac{\partial}{\partial x_6}
\]
Each $\mathcal{N}_i$ is a  non-decomposable generalized 4-Poisson tensor.
Moreover, with the pair of Hamiltonians obtained by \emph{removing} $h_i$, we get
\[
\Gamma(f) \;=\; \mathcal{N}_i\big(df, d h_j, d h_3,dh_4\big), \qquad \{i,j\}=\{1,2\}.
\]
Indeed, since $d h_1=dx_1$ and $d h_2=dx_2$, a direct contraction yields
\[
\mathcal{N}_1(df,dx_2,dx_3,dx_4)=\frac{\partial f}{\partial x_7},\quad
\mathcal{N}_2(df,dx_1,dx_3,dx_4)=\frac{\partial f}{\partial x_7},
\]
so in two cases the associated generalized Hamiltonian vector field is precisely $\Gamma$.
\end{example}
\medskip

\subsection{1-angle integrable system with three actions}

Consider an integrable Hamiltonian system with one angle variable $\theta$ and five independent action-like
constants $I_1,...,I_5$ on Manifold $\mathbb R^4\times \mathbb S^1$. Suppose the dynamics is a uniform drift in $\theta$,
\[
\dot{I}_1=\cdots=\dot I_5=0,\qquad \dot{\theta}=\omega_0,\quad \omega_0\in\mathbb{R}\setminus\{0\}.
\]
Identify
\[
(x_1,\cdots,x_6)=(I_1,\cdots,I_5 ,\theta),
\]
so that the vector field is
\[
\Gamma=\frac{d}{dt}=\omega_0\,\frac{\partial}{\partial \theta}=\omega_0\,\frac{\partial}{\partial x_6}.
\]
Rescaling time by $t\mapsto \omega_0 t$ (or, equivalently, rescaling the Nambu tensor by a constant factor)
gives precisely  $\Gamma=\partial_{x_6}$ and the five commuting translational symmetries
$X_i=\partial_{x_i}$, $i=1,...,5$, with first integrals $h_i=x_i=I_i,i=1,2,3$.
Hence the non-decomposable 
 generalized 4-Poisson tensor
\[
\mathcal N=\Gamma\wedge X_1\wedge X_2\wedge X_3+\Gamma\wedge X_1\wedge X_4\wedge X_5
\]
reproduces the dynamics via
\[
\frac{df}{dt}=\Gamma(f)=N\big(df,dh_1,dh_2,dh_3\big).
\]

\begin{remark}
To model a nonconstant frequency 
\[
\dot{\theta} = \omega(I_1, \ldots, I_5),
\]
we simply replace \(\mathcal{N}\) by 
\[
\mathcal{N}_\rho = \rho(x_1, \ldots, x_5)\, \partial_{x_5} \wedge \partial_{x_1} \wedge \partial_{x_2} \wedge \partial_{x_3} + \rho(x_1, \ldots, x_5)\, \partial_{x_5} \wedge \partial_{x_1} \wedge \partial_{x_4} \wedge \partial_{x_5},
\]
and set \(\rho = \omega\). Since the degree of the Schouten bracket \([\mathcal{N}_\rho, \mathcal{N}_\rho]\) exceeds the dimension of the manifold \(\mathbb{R}^6\), we conclude that \([\mathcal{N}_\rho, \mathcal{N}_\rho] = 0\). Hence, \(\mathcal{N}_\rho\) defines a non-decomposable generalized 4-Poisson tensor. Moreover, since 
\[
\Gamma = \mathcal{N}_\rho(\,\cdot\,, dh_1, dh_2, dh_3) = \omega\, \partial_{x_5}
\]
and the functions \(h_i\) remain integrals, it follows that \(\Gamma\) is a generalized Hamiltonian vector field with respect to the generalized 4-Poisson tensor \(\mathcal{N}_\rho\).
\end{remark}

Some representative examples fitting this structural template include:

\begin{itemize}
\item \textbf{One-phase limit of an integrable mode}: A decoupled mode with phase variable \(\theta\) evolving at a possibly nonconstant frequency, while five action-like invariants \(I_1, I_2, I_3, I_4, I_5\) remain fixed. This can be viewed as a reduced limit of a higher-dimensional integrable system.
  
\item \textbf{Reduced symmetric rigid body or rotor}: A rigid body spinning around a symmetry axis, with \(\theta\) representing the spin angle. Due to the symmetry, five combinations of angular momentum components are conserved, forming the invariants \(I_i\). The reduced system evolves in \(\theta\) while the \(I_i\) serve as parameters.
\end{itemize}

\begin{remark}
Allowing a frequency of the form
\[
\dot{\theta} = \omega(I_1, I_2, I_3, I_4, I_5)
\]
corresponds to replacing the constant-coefficient, non-decomposable generalized Poisson tensor \(N\) with
\[
\mathcal{N}_\rho = \rho(x_1, \ldots, x_5)\, \partial_{x_5} \wedge \partial_{x_1} \wedge \partial_{x_2} \wedge \partial_{x_3} + \rho(x_1, \ldots, x_5)\, \partial_{x_5} \wedge \partial_{x_1} \wedge \partial_{x_4} \wedge \partial_{x_5}.
\]
Since the Schouten bracket \([\mathcal{N}_\rho, \mathcal{N}_\rho]\) has degree higher than the dimension of the ambient manifold \(\mathbb{R}^6\), it vanishes identically. Hence, \(\mathcal{N}_\rho\) defines a valid non-decomposable generalized 4-Poisson tensor.

The dynamics is then governed by the generalized Hamiltonian vector field
\[
\Gamma = \mathcal{N}_\rho(\,\cdot\,, dh_1, dh_2, dh_3) = \omega\, \partial_{x_5},
\]
and the functions \(h_i = x_i = I_i\) remain conserved. Therefore, the generalized Hamiltonian structure under a non-decomposable generalized Poisson tensor remains valid even when the frequency depends nontrivially on all five invariants.
\end{remark}

\subsection{4D cyclic Lotka--Volterra with spectators: explicit quasi-Hamiltonian form}

%{\color{blue} Here we have the problem of notation, in which we are using $X_i$ and $Y_i$ with a different meaning of the introduced concept before. Shall we rename $Y_i$ as $\tilde{X}_i$?}\textcolor{red}{I think it’s quite clear as it is now, so there’s no need to rename $Y_i$ as $\tilde{X}_i$.}

Consider the 4D cyclic Lotka--Volterra system
\[
\begin{aligned}
\dot x_1 &= x_1(x_2 - x_4),\\
\dot x_2 &= x_2(x_3 - x_1),\\
\dot x_3 &= x_3(x_4 - x_2),\\
\dot x_4 &= x_4(x_1 - x_3).
\end{aligned}
\]
Its vector field is
\[
\Gamma \;=\; \Gamma_1\,\partial_{x_1}+\Gamma_2\,\partial_{x_2}+\Gamma_3\,\partial_{x_3}+\Gamma_4\,\partial_{x_4},
\quad
\begin{cases}
\Gamma_1=x_1(x_2-x_4),\\
\Gamma_2=x_2(x_3-x_1),\\
\Gamma_3=x_3(x_4-x_2),\\
\Gamma_4=x_4(x_1-x_3).
\end{cases}
\]

Two first integrals are
\[
h_\Sigma=x_1+x_2+x_3+x_4,\qquad h_\mathcal{N}=x_1x_2x_3x_4,
\]
so $\Gamma(h_\Sigma)=0$ and $\Gamma(\log h_\mathcal{N})=0$.
The Euler (scaling) symmetry
\[
X_s \;=\; x_1\partial_{x_1}+x_2\partial_{x_2}+x_3\partial_{x_3}+x_4\partial_{x_4}
\]
commutes with $\Gamma$, i.e.\ $[X_s,\Gamma]=0$, and
\[
(\Gamma\wedge X_s)(d h_\Sigma)=-\,h_\Sigma\,\Gamma.
\]

The additional spectator variables $y_1,y_2$ are coordinates that evolve trivially, i.e.\ their values remain constant along any trajectory. 
Spectators can be interpreted as {environmental labels} (such as a fixed habitat index, external condition, or conserved resource level) or as {parameters promoted to dynamical variables} that remain constant on the relevant timescale. 

Adjoin two spectators $y_1,y_2$ with trivial dynamics
\[
\dot y_1=0,\qquad \dot y_2=0,
\]
and let
\[
Y_1=\partial_{y_1},\qquad Y_2=\partial_{y_2}.
\]
They commute with everything: $[Y_i,\Gamma]=[Y_i,X_s]=[Y_1,Y_2]=0$.
We take the “Hamiltonian” function $h:=h_\Sigma$, for which
\[
\Gamma(h)=Y_1(h)=Y_2(h)=0,\qquad X_s(h)=h_\Sigma.
\]

Define
\[
\mathcal{N} \;=\; \Gamma\wedge Y_1 \;+\; \Gamma\wedge X_s \;+\; Y_1\wedge Y_2.
\]
Its full coordinate expansion on $(x_1,x_2,x_3,x_4,y_1,y_2)$ is
\[
\mathcal{N} \;=\;
\sum_{i=1}^4 \Gamma_i\,\partial_{x_i}\wedge\partial_{y_1}
\;+\;
\sum_{1\le i<j\le 4} B_{ij}\,\partial_{x_i}\wedge\partial_{x_j}
\;+\;
\partial_{y_1}\wedge\partial_{y_2}.
\]

with coefficients
\[
\begin{aligned}
&\Gamma_1=x_1(x_2-x_4),\quad
\Gamma_2=x_2(x_3-x_1),\quad
\Gamma_3=x_3(x_4-x_2),\quad
\Gamma_4=x_4(x_1-x_3),\\[2mm]
&B_{12}=\Gamma_1 x_2-\Gamma_2 x_1 \;=\; x_1x_2\,(x_1+x_2-x_3-x_4),\\
&B_{13}=\Gamma_1 x_3-\Gamma_3 x_1 \;=\; 2\,x_1x_3\,(x_2-x_4),\\
&B_{14}=\Gamma_1 x_4-\Gamma_4 x_1 \;=\; x_1x_4\,(-x_1+x_2+x_3-x_4),\\
&B_{23}=\Gamma_2 x_3-\Gamma_3 x_2 \;=\; x_2x_3\,(-x_1+x_2+x_3-x_4),\\
&B_{24}=\Gamma_2 x_4-\Gamma_4 x_2 \;=\; 2\,x_2x_4\,(x_3-x_1),\\
&B_{34}=\Gamma_3 x_4-\Gamma_4 x_3 \;=\; x_3x_4\,(-x_1-x_2+x_3+x_4).
\end{aligned}
\]

For $h=h_\Sigma$,
\[
\mathcal{N}(\,\cdot\,, d h) \;=\; X_s(h)\,\Gamma \;=\; h_\Sigma\,\Gamma,
\]
so the LV dynamics is quasi-Hamiltonian with multiplier $h_\Sigma$ (a first integral). The term $Y_1\wedge Y_2$ together with $\Gamma\wedge Y_1$ makes $\mathcal{N}$ non-decomposable.

\section{Summary}

In this work, we have developed a systematic geometric framework that unifies Hamiltonian, quasi-Hamiltonian, and Nambu--Poisson formulations through the use of generalized higher-order Poisson tensors. The central result demonstrates that families of commuting symmetries and conserved quantities naturally induce non-decomposable generalized Poisson tensors, under which a dynamical system can be represented as a quasi-Hamiltonian or genuinely Hamiltonian system. This construction extends previous symmetry-based approaches, such as Hojman’s method, into the higher-order setting, where Nambu-type and generalized Poisson structures coexist within a common geometric scheme.

We have established explicit theorems guaranteeing the existence of quasi-generalized Hamiltonian realizations for systems admitting sufficiently many commuting symmetries, and we have provided concrete examples illustrating these constructions, including multi-angle systems and generalized Lotka--Volterra dynamics. The appearance of non-decomposable generalized Poisson tensors reveals a class of systems that, while not Nambu--Poisson in the strict sense, still possess Hamiltonian-type representations governed by higher-order brackets.

These results highlight the unifying role of Nambu--Poisson geometry in the study of symmetry-induced structures and open the way for further investigations into integrability, reduction procedures, and quantization within the framework of higher-order geometric mechanics.

	\section*{Acknowledgments}
	The research of the second author is supported by
	NSFC (Grant No.
	12401234). 
	%Special Funds of Provincial Industrial Innovation of Jilin Province China (Grant No.
	%	2017C028-1), Project of Science and Technology Development of Jilin Province China (Grant
	%	No. 20190201302JC),  National Basic Research Program of China (Grant No. 2013CB834100), NSFC (Grant No.
	%	12071175), NSFC (Grant No.
	%	12471183).%The author expresses his deep gratitude to anonymous referees for their valuable comments which have improved the paper.
	$\\$
	
	\noindent$\mathbf{Conflict\;of\;interest\;statement.}$ On behalf of all authors, the corresponding author states that there is no conflict of interest.
	
	$\\$
	\noindent$\mathbf{Data\;availability.}$ Data sharing is not applicable to this article as no new data were created or analyzed in this study.

\end{document}